\documentclass[journal]{IEEEtran}
\IEEEoverridecommandlockouts
\usepackage{cite}
\usepackage{amsmath,amssymb,amsfonts}
\usepackage{graphicx}
\usepackage{textcomp}
\usepackage{xcolor}
\usepackage{amsmath, graphics,amssymb,epsfig,subfigure,amsthm,mathtools}
\usepackage{cite}
\usepackage{color}
\usepackage{url}
\usepackage{color,soul}
\usepackage{optidef}
\usepackage{stfloats}
\usepackage{algorithm}
\usepackage{mathrsfs}
\usepackage{algpseudocode}

\usepackage{mathtools, nccmath}

\newtheorem{lemma}{Lemma}

\DeclareMathOperator{\Tr}{Tr}
\DeclareMathOperator{\rank}{rank}
\DeclareMathOperator{\diag}{diag}
\DeclareMathOperator{\vect}{vec}

\newcommand{\bea}{\begin{eqnarray}}
\newcommand{\eea}{\end{eqnarray}}
\newcommand{\bdp}{\begin{displaymath}}
\newcommand{\edp}{\end{displaymath}}
\usepackage{graphicx,amsmath}

\makeatletter
\def\hlinewd#1{%
\noalign{\ifnum0=`}\fi\hrule \@height #1 \futurelet \reserved@a\@xhline}

\newcommand{\RN}[1]{%
  \textup{\uppercase\expandafter{\romannumeral#1}}%
}

\makeatother

\def\BibTeX{{\rm B\kern-.05em{\sc i\kern-.025em b}\kern-.08em
    T\kern-.1667em\lower.7ex\hbox{E}\kern-.125emX}}
\begin{document}

\title{Detection with Uncertainty in Target Direction for Dual Functional Radar and Communication Systems\\
%\thanks{Identify applicable funding agency here. If none, delete this.}
}

\author{Mateen Ashraf, Anna Gaydamaka, Dmitri Moltchanov, John Thompson, Mikko Valkama, Bo Tan\\
\thanks{M. Ashraf, A. Gaydamaka, D. Moltchanov, B. Tan, and M. Valkama are with the Faculty of Information Technology and Communication Sciences, Tampere University, Finland. (e-mails:  
%\IEEEauthorblockB{\textit{Institute for Digital Communications, School of Engineering}, The University of Edinburgh, UK} \\
% JST - some small updates to my affiliation here.
\{mateen.ashraf, anna.gaydamaka, dmitri.moltchanov, bo.tan, mikko.valkama\}@tuni.fi.)}
\thanks{J.Thompson is with the Institute for Digital Communications, School of Engineering, The University of Edinburgh, UK. (email: john.thompson@ed.ac.uk)}
}
\maketitle

\begin{abstract}
Dual functional radar and communication (DFRC) systems are a viable approach to extend the services of future communication systems. Most studies designing DFRC systems assume that the target direction is known. In our paper, we address a critical scenario where this information
is not exactly known. For such a system, a signal-to-clutter-plus-noise ratio (SCNR) maximization problem is formulated. Quality-of-service constraints for communication users (CUs) are also incorporated as constraints on their received signal-to-interference-plus-noise ratios (SINRs). To tackle the nonconvexity, an iterative alternating optimization approach is developed where, at each iteration, the optimization is alternatively performed with respect to transmit and receive beamformers. Specifically, a penalty-based approach is used to obtain an efficient sub-optimal solution for the resulting subproblem with regard to transmit beamformers. Next, a globally optimal solution is obtained for receive beamformers with the help of the Dinkleback approach. The convergence of the proposed algorithm is also proved by proving the nondecreasing nature of the objective function with iterations. The numerical results illustrate the effectiveness of the proposed approach. Specifically, it is observed that the proposed algorithm converges within almost $3$ iterations, and the SCNR performance is almost unchanged with the number of possible target directions.
\end{abstract}


\section{Introduction}

\IEEEPARstart{I}{n} recent years, integrated sensing and communication systems (ISAC) have received much attention. These investigations are mainly driven by the scarcity of available frequency resources and the ever-increasing adjunct services that are becoming an integral part of the communication systems. Specifically, services such as connected autonomous systems, autonomous driving, vehicle-to-everything (V2X) communications, human activity sensing, and unmanned aerial vehicle (UAV) networks are just a few examples that require the sensing functionality built into communication systems for their proper functioning. Given the impetus of these advances, it is inevitable to envision a future 5G/6G communication system without the built-in sensing functionality \cite{IMT}. 

A very common approach to design an ISAC system is through the dual function radar communication (DFRC) systems, which combine both radar and communication systems through shared use of frequency resources and hardware circuitry. Owing to their spectrum and infrastructure efficiency, DFRC systems have been the focus of several research studies that aim at characterizing and/or optimizing the performance of DFRC systems. A major difference observed when sensing functionality is integrated with communication systems is that performance metrics also differ from those of conventional communication systems. The most commonly used performance metrics for sensing localization are based on the Cr\'amer-Rao bound (CRB) of the distance, angle or speed parameters. Specifically, for localization purposes, the objective is to devise transmission schemes such that the performance is as close to the CRB as possible. On the other hand, for detection purposes, the goal is to allocate transmission resources so that the detection probability is maximized while maintaining the performance of the communication system above a given threshold. 

Regarding localization performance, recent works \cite{sxu, yrong, fliu,ivali, lx} provided CRB minimization schemes under different system setups. In \cite{sxu}, a hybrid approach with known/unknown placements of multiple sensors is presented to estimate the location of multiple targets. The design of various detectors based on the minimization of CRB is presented in \cite{yrong}. The works in \cite{sxu}, \cite{yrong} do not consider the joint operation of sensing and communication. To fill this void, \cite{fliu} provides a CRB minimization scheme. Specifically, CRB is used as a performance metric for target direction estimation, and then a CRB minimization beamforming design is proposed, which also guarantees a predefined level of signal-to-interference-plus-noise ratio (SINR) for each communication user (CU). To account for the energy efficiency, an $l_0$ norm-based antenna selection strategy was proposed in \cite{ivali} for minimizing the CRB. On the other hand, an antenna sharing scheme for sensing and communication/computing purposes was proposed in \cite{lx}. The beamforming design was then provided to minimize the mean squared error for the computing functions.

It is well known that the detection probability is directly proportional to the signal-to-clutter plus noise ratio (SCNR) \cite{dm}. To this end, several interesting approaches have been developed to maximize SCNR in DFRC systems. An SCNR maximization scheme was proposed in \cite{co} for a single input single output (SISO) orthogonal frequency division multiplexing (OFDM) system with a single CU and a target. However, in this work, it was assumed that the interference caused by clutter is independent of the transmitted/reflected signal. Specifically, clutter interference was assumed to follow a Gaussian distribution. This assumption is generally not true since the clutter also reflects the transmitted signal in a similar fashion as the desired target, and hence should be treated accordingly when designing the transmission scheme. To avoid this issue, a multi-user DFRC system was considered in \cite{lc}. In particular, beamforming designs were proposed with the possibility of dedicated and non-dedicated probing
radar signals. It was concluded that when the CUs have knowledge of the radar signal, it is beneficial to use the dedicated probing signal. On the other hand, it was shown in \cite{Ashraf1} that when CUs do not possess the knowledge of the dedicated radar signal, it is not necessary to use a dedicated radar probing signal. Although these systems treated interesting theoretical system models, an important aspect of the target arrival direction was not explored in these works. In these works, it was assumed that the arrival direction of the target is exactly known apriori, and then transmit/receive beamformers were designed based on that particular direction. However, in real-world scenarios, exact knowledge of the target direction may not be available. Only a set of possible arrivals directions may be known in advance. Therefore, it is important to investigate the DFRC system under uncertainty in the target arrival direction.

\subsection{Related Work}

In the context of multi-target sensing, several interesting approaches have been developed recently. Specifically, for multi-antenna DFRC systems, the majority of research work is related to designing beamformers that can satisfy the sensing and communication requirements of modern ISAC systems. To achieve this goal, several interesting optimization problems are solved in the literature. A trade-off between the sidelobe and mainlobe power distribution within the DFRC system was explored in \cite{Lchen1}, and a Pareto boundary analysis was carried out to identify the corresponding operational regions. A beam pattern matching approach was discussed in \cite{Fan1}. Specifically, the goal of this work was to minimize the discrepancy between the transmitted beam pattern and the desired radar beam pattern while satisfying the SINR requirements of the CUs. In \cite{elder} it was identified that the degrees of freedom (DoF) for the approach proposed in \cite{Fan1} is limited by the number of CUs and can be effectively smaller than the number of transmit antennas. Therefore, a new beamforming technique was devised in \cite{elder} that guarantees a higher DoF compared to that achieved in \cite{Fan1}. 

While the schemes proposed in \cite{Fan1, elder} aim to minimize the dissimilarity between the achieved beam pattern and the ideal radar beam pattern while meeting the SINR constraints of the CUs, a more generic objective function for ISAC performance was formulated in \cite{Fan2}. Specifically, a weighted sum of multi-user communication interference (MUI) and the radar beampattern dissimilarity metric were minimized. This work was then extended to MIMO OFDM systems in \cite{Fan3} by incorporating the peak-to-average power ratio (PAPR) constraint in the problem formulation. 

The beampattern matching schemes discussed above focus only on the transmit beamformers. In these schemes, all antennas at the DFRC transmission stations are dedicated to transmission. The idea then is to illuminate certain angular regions with the help of an attained beampattern. However, another possible implementation of the DFRC assumes the division of total antennas into transmit and receive antennas. The function of the receive antennas is to catch the echoes reflected by the target. Therefore, the design of the receive beamformer is also a crucial aspect and should be taken into account. In this context, \cite{Palomar} considered a SCNR maximization problem with PAR constraints. Specifically, the majorization-minimization technique was applied to obtain the transmit beamformer, while closed-form expressions were obtained for the receive beamformers. However, the system model considered in \cite{Palomar} did not consider the QoS constraints for CUs. To fill this gap, a joint optimization with respect to transmit and receive beamformers was proposed in \cite{Ashraf2}. In particular, an iterative alternating optimization-based algorithm was proposed to solve the SCNR maximization problem. However, the convergence guarantee was not proved. On the other hand, an iterative algorithm with convergence guarantee was proposed in \cite{Cwen}. In this work, the author devised an alternating optimization scheme, and for fixed values of CUs transmit beamformers, multiple receive beamformers were designed according to the target directions. Therefore, the number of receive beamformers, according to this work, is equal to the number of radar targets. Another important observation regarding \cite{Cwen} is that the information symbols of the CUs are parameters of the optimization problem. This means that the optimization problem needs to be solved for each new set of transmission symbols. This drastically increases the computational complexity of the proposed solution in \cite{Cwen}, as the frequency of change in information symbols is directly proportional to the data rate.

\subsection{Our Contributions}

This paper considers an ISAC system with multiple CUs and a target with ambiguity in arrival direction. The ultimate goal is to maximize the detection probability of the radar whilst satisfying the minimum data rate requirements of the CUs for all the possible target directions. In this paper, we propose an iterative algorithm based on alternating optimization for SCNR maximization in DFRC systems with uncertainty in the target arrival direction. Unlike \cite{Cwen}, this work proposes using a single receive beamformer at the DFRC station and the parameters of the optimization problem include the channel state information of the CUs and the target/clutter directions. The implication of this is that the receiving side circuitry at the DFRC station is less complex compared to the scheme proposed in \cite{Cwen}. To solve the SCNR maximization problem, we make use of the penalty-based approach and the successive convex approximation (SCA) technique in each iteration to obtain the optimal transmit beamformers for CUs. Next, by fixing the transmit beamformer, it is observed that the resulting optimization problem is a special case of the generalized fractional program for which the global optimal solution can be achieved. These facts and the boundedness of the objective value are then used to show that the objective value is nondecreasing in each iteration, and hence the convergence is guaranteed. 

The main contributions of this paper are listed below.
\begin{itemize}
    \item An alternating optimization-based iterative algorithm is proposed for solving the SCNR maximization problem in DFRC systems where the target arrival direction is ambiguous. Specifically, the proposed algorithm alternately maximizes the SCNR with respect to the transmit beamformers of the CUs and the receive beamformer of the receiver. In contrast to the previous work, the proposed method uses a single receive beamformer to combine the echoes reflected from the target. This helps to reduce the complexity of the receive circuitry in the DFRC receive section.
    \item The convergence of the proposed iterative algorithm is also proved. In particular, the surrogate function properties are used to establish that the objective function is nondecreasing with respect to the optimization of transmit beamformers and the guarantee of global optimal solution with respect to receive beamformer is used to establish that the objective function is non-decreasing as the iterations progress.  
    \item Numerical results are included to illustrate the effectiveness of the proposed algorithm and validate the theoretical claim about convergence. Numerical results indicate that the SCNR performance remains stable with an increase in uncertainty in the target arrival direction. Moreover, the SCNR performance converges after $3$ iterations.
\end{itemize}

\textit{Notations:} The important notation used in this paper is given below. Sets of complex, real numbers are represented by $\mathbb{C}, \mathbb{R}$, respectively. Vectors and matrices are represented by bold, small, and capital letters, respectively. $\bold{A}^H, \bold{A}^{-1}, \Tr(\bold{A}), \rank(\bold{A})$ are the hermitian transpose, inverse, trace and rank of the complex matrix $\bold{A}$, respectively. $\bold{A}\succeq \bold{0}$ means that $\bold{A}$ is positive semidefinite. The Kronekcer product is represented by $\otimes$, and the operator $\diag(\bold{x})$ produces a square matrix whose diagonal entries are given by the vector $\bold{x}$. The vector $\bold{e}_i$ denotes a vector with zero elements except the $i$-th element, which is $1$. The identity matrix is represented by $\bold{I}$, where the dimensions are chosen appropriately. The magnitude of a scalar $a$ and a vector $\bold{x}$ are represented by $|a|, \|\bold{x}\|$, respectively. $E(x)$ denotes the expectation of the random variable $x$.

The rest of this paper is organized as follows. In Section II, we discuss the system model and problem formulation. The proposed approach is discussed in Section III. Section IV presents the proposed algorithm. The numerical results are discussed in Section V. Finally, the conclusions of the paper are presented in Section VI.

\section{System Model and Problem Formulation}

In this section, we first specify the system model and describe the environment of interest. Then we formulate the problem and, finally, describe the existing solution approaches.

\subsection{System Design and Environment}

This subsection presents the main parameters of the system, the underlying assumptions, and important performance metrics for the radar and communication system. In this paper, as shown in Fig. 1, we assume an ISAC transceiver equipped with $N$ element uniform linear array (ULA) serves $K$ single antenna CUs. In addition to serving the CUs, we assume that a target must also be detected which can be at an angle $\theta_T \in \{\theta_1,\theta_2,\cdots, \theta_I\}$ from the ISAC transmitter\footnote{While the possible angle of arrivals is assumed to be members of the discrete set, the proposed approach can also be extended to the case where the possible angle of arrivals are members of a continuous set.}. Additionally, we assume that there are $J$ clutters that also reflect the signal and cause signal-dependent interference at the ISAC receiver.

\begin{figure}[t]
\begin{center}
\includegraphics[width=\columnwidth]{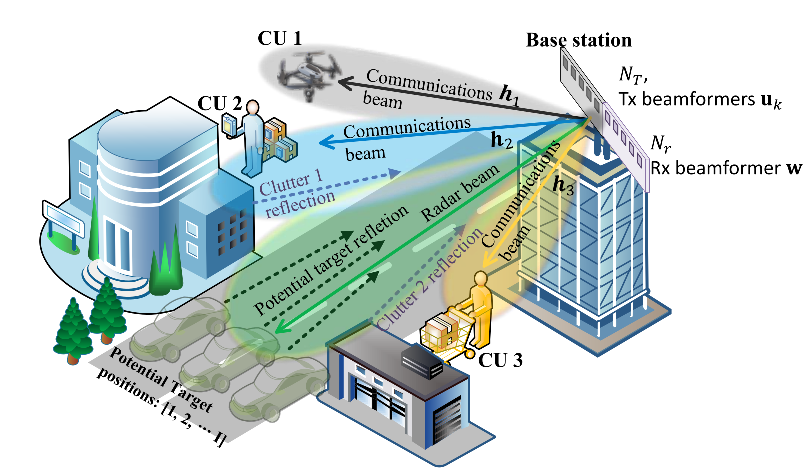}
\end{center}
\vspace{-4mm}
\caption{The presumed multiuser downlink integrated communications and sensing scenario with $K=3$ CUs and $I=3$ possible directions of target. The light shade of the targets indicate the uncertainty in the target direction.}
\label{fig:sys_model}
\vspace{-0mm}
\end{figure}

From the available $N$ antennas at the ISAC transceiver station, we assume that the $N_t$ antennas are used to transmit information to the CUs while it uses the $N_r$ antennas to receive reflection from the target to detect the target. The beamforming vector and the information symbol for $k$-th CU are denoted by $\bold{u}_k \in \mathbb{C}^{N_t \times 1}$, and $s_k$, respectively. In the following, we assume $E[s_k]=0,E[|s_k|^2]=1, E[s_js_k^*]=0,$ for $j \neq k$.

The overall transmitted symbol from the ISAC transmitter can be written as
\bea
\bold{x}=\sum_{k=1}^K\bold{u}_ks_k.
\eea

Furthermore, the channel vector between the transmitter and $k$-th CU is represented by $\bold{h}_k \in \mathbb{C}^{N_t \times 1}$, which is assumed to be known at the transmitter \cite{Lchen1,lc,Ashraf1}. Next, we present the performance metrics for communication and radar systems, respectively.
\subsubsection{Communication System Performance Metric}
With the above assumptions, the received signal at the $k$-th user can be written as 
\bea
y_k = \bold{h}_k^H\bold{u}_ks_k+\sum_{i=1,i\neq k}^K\bold{h}_k^H\bold{u}_is_i+\omega_k,
\eea
where $\omega_k\in \mathbb{C}$ is the additive white Gaussian noise (AWGN) at CU $k$ with mean zero and variance $N_0$. 

Then, the corresponding signal-to-interference-plus-noise ratio (SINR) at the $k$-th CU for first scenario is given as 
\bea
\gamma_k=\frac{|\bold{h}_k^H\bold{u}_k|^2}{\sum_{i=1,i\neq k}^K|\bold{h}_k^H\bold{u}_i|^2+N_o}.
\eea

According to the Shannon formula, the communication rate is an increasing function of the SINR. Therefore, to have satisfactory communication performance, we assume that the $k$-th CU requires its SINR to be at least $\Gamma_k$. Mathematically, this requirement can be represented as 
\bea
\gamma_k=\frac{|\bold{h}_k^H\bold{u}_k|^2}{\sum_{i=1,i\neq k}^K|\bold{h}_k^H\bold{u}_i|^2+N_o} \geq \Gamma_k. \label{EE1}
\eea

\subsubsection{Radar System Performance Metric}

In this subsection, we present the radar system performance metric for two possible scenarios. In the system model considered, we assume that the target is located at an angle $\theta_i^T$ from the ISAC transmitter, and the $j$-th clutter component is located at an angle $\theta_j^C$. We further assume that the values of $\theta_i^T,\theta_j^C$ are known apriori, an assumption widely used in the literature see for example \cite{gc,lc,Lchen1,Ashraf1,Cwen} and reference therein. Then, the signal received on the radar can be written as
\bea
\bold{r}(\theta_i^T)\! =\! \alpha_i^T \bold{a}_r(\theta_i^T) \bold{a}_t^H(\theta_i^T) \bold{x}\! +\! \sum_{j=1}^J\alpha_i \bold{a}_r(\theta_j^C) \bold{a}_t^H(\theta_j^C) \bold{x}\! +\! \bold{n}, \nonumber
\eea
which can be simplified to
\bea
\bold{r}(\theta_i^T)=\alpha_i^T \bold{A}(\theta_i^T)\bold{x}+\sum_{j=1}^J\alpha_j^C \bold{A}(\theta_j^C)\bold{x} + \bold{n},
\eea
where $\bold{A}(\theta)=\bold{a}_r(\theta) \bold{a}_t^H(\theta)\in \mathbb{C}^{N_r \times N_t}$, $\alpha_i^T, \alpha_j^C \in \mathbb{C}$ are the complex channel between target and radar, and between clutter and radar, $\bold{n}\in \mathbb{C}^{N_r \times 1}$ is additive white Gaussian noise (AWGN) with $\bold{n} \sim \mathcal{CN}(\bold{0}, \sigma_r\bold{I})$, $\bold{a}_z(\theta)\in \mathbb{C}^{N_z \times 1}$ is the transmit are receive steering vectors for $z \in \{t,r\}$, respectively. The dependencies of the steering vectors $\bold{a}_t(\theta), \bold{a}_r(\theta)$ on the angle $\theta$ are given as
\bea
\bold{a}_t(\theta) = [1, e^{-j 2 \pi \Delta \sin(\theta)}, \cdots, e^{-j 2 \pi (N_t-1)\Delta \sin(\theta)}]^H,
\eea
\bea
\bold{a}_r(\theta) = [1, e^{-j 2 \pi \Delta \sin(\theta)}, \cdots, e^{-j 2 \pi (N_r-1)\Delta \sin(\theta)}]^H,
\eea
where $\Delta=\frac{\lambda}{2}$ and $\lambda$ is the carrier wavelength. Therefore, to obtain the value of the matrix $\bold{A}(\theta)$, the only information required is $\theta$.

We assume that $\alpha_i^T, \alpha_j^C$ are independently distributed from $\bold{h}_k$'s. After reception, the radar performs receive beamforming (combining) with vector $\bold{w}$ on the received signal, then the output of the radar receiver is given as 
\bea
y_r(\theta_i^T) = \bold{w}^H\bold{r}(\theta_i^T).
\eea

Subsequently, the average radar signal-to-clutter-plus-noise ratio (SCNR) can be written as 
\begin{equation}
\small
\gamma_r(\theta_i^T) = \frac{|\alpha_i^T|^2E[|\bold{w}^H \bold{A}(\theta_i^T) \bold{x}|^2]}{E\left[\bold{w}^H\left(\sum_{j=1}^J|\alpha_j^C|^2\bold{A}(\theta_j^C)\bold{x}\bold{x}^H\bold{A}^H(\theta_j^C)+\sigma_r\bold{I}\right)\bold{w}\right]}, 
\end{equation}
\normalsize
which, after performing the expectation operation, can be written in compact form as
\bea
\small
\bar{\gamma}_r(\theta_i^T) = \frac{\sum_{k=1}^K|\alpha_i \bold{w}^H \bold{A}(\theta_i) \bold{u}_k|^2}{\sum_{j=1}^J\sum_{k=1}^K|\alpha_j\bold{w}^H\bold{A}(\theta_j)\bold{u}_k|^2+\sigma_r\bold{w}^H\bold{w}}. \label{E1}
\eea
\normalsize

It is clear from (\ref{EE1}), (\ref{E1}) that the communication and radar performance depend on the beamforming vectors $\bold{u}_k$'s and $\bold{w}$. In the next subsection, we formulate an optimization problem to find the optimal values of $\bold{u}_k$'s and $\bold{w}$ while considering radar and communication performance simultaneously.

\subsection{Problem Formulation}

In this paper, we are interested in maximizing the detection probability of the radar. As noted above, the radar detection probability is directly proportional to the radar SCNR. Therefore, our aim is to maximize the radar SCNR while satisfying the communication requirements of the CUs. Overall, the mathematical formulation of the optimization problem to find the appropriate beamforming vectors is given as follows.

\textbf{P1:}
\begin{maxi}|s|
{\bold{u}_k,\bold{w}} {\min_{\theta_i^T \in \Theta^T} \gamma_r(\theta_i^T,\bold{w})}{}{}
\addConstraint{C1:~\frac{|\bold{h}_k^H\bold{u}_k|^2}{\sum_{i=1,i\neq k}^K|\bold{h}_k^H\bold{u}_i|^2+N_o}\geq \Gamma_k}
 \addConstraint{C2:~\sum_{k=1}^K \|\bold{u}_k\|^2\leq P_{max}}
 \addConstraint{C3:~ \|\bold{w}\|^2=1,}
\end{maxi} 
where $\Theta^T=\{\theta_1^T,\theta_2^T,\cdots,\theta_I^T\}$ is the set that contains the possible angles of the target with respect to the ISAC transmitter.

In \textbf{P1}, the objective is to maximize the average SCNR of the radar system while considering the ambiguity about the angle of the target. Since the data rate is directly proportional to SINR, constraints $C1$ guarantee that the minimum data rate requirements of the CUs are met. The constraint $C2$ ensures that the total transmitted power is no more than the maximum transmitted power allowed. Finally, $C3$ imposes a constraint on the receive beamformer. It is clear that \textbf{P1} is a nonconvex optimization problem. The nonconvexity is caused by the nonconvex objective function as well as the non-convex constraint functions within $C1$, and $C3$. More specifically, constraint $C1$ requires that a convex function is greater than another convex function which leads to nonconvexity. On the other hand, the constraint $C3$ requires that the convex function be equal to a constant.

In the following subsection, we highlight some critical issues within existing approaches that can be used to solve special cases of \textbf{P1}.
\subsection{Existing Approaches and the Corresponding Issues}
In order to address the difficulty posed by the problem \textbf{P1}, an iterative approach was initially proposed in \cite{gc} for radar systems. This scheme was then extended in \cite{lc,Ashraf1} to incorporate communication systems with only a single possible direction of arrival of the target. In this approach, the main idea is to simplify the objective function by ignoring the dependence of some part, specifically the part introduced within the objective function due to the receive beamformer mathematical expression, of the objective function on the underlying waveform. In the following, we illustrate one major issue with regard to convergence of such an approach.

Note that when the target direction is already known, the SCNR maximization problem can be written as \cite{lc,Ashraf1}
\begin{maxi}|s|
{\bold{u}_k} {|\alpha_0|^2\sum_{k=1}^K \bold{u}_k^H \bold{A}^H(\theta_0)\bold{G}^{-1}\left(\{\bold{u}_k\}\right)\bold{A}(\theta_0)\bold{u}_k}{}{}
\addConstraint{C1, C2,}
\end{maxi}
where
\bea
\bold{G}\left(\!\{\bold{u}_k\}\!\right)\!=\! \sum_{j=1}^J|\alpha_j|^2\bold{A}(\theta_j)\!\left(\sum_{k=0}^K\bold{u}_k\bold{u}_k^H\!\right)\!\bold{A}^H(\theta_j)+\bold{I}.
\eea

Then, the main idea proposed in \cite{gc} is to ignore the dependence of $\bold{G}$ on $\{\bold{u}_k\}$'s and solve the following simplified optimization problem in the $m$-th iteration.
\begin{maxi}|s|
{\bold{u}_k} {|\alpha_0|^2\sum_{k=1}^K \bold{u}_k^H \bold{A}^H(\theta_0)\bold{G}^{-1}\left(\{\bold{u}_k^{m-1}\}\right)\bold{A}(\theta_0)\bold{u}_k}{}{}
\addConstraint{C1, C2,}
\end{maxi}
where $\{\bold{u}_k^{m-1}\}$ is the optimal solution obtained by solving (14) during the $(m-1)$-th iteration. By doing so, the inner matrix $\bold{G}^-1\left(\{\bold{u}_k^{m-1}\}\right)$ becomes a constant parameter of the optimization problem instead of being a function of the optimization variable. Although the simplified problem is still nonconvex, it can be solved through the SDR technique, and it has been proven that the resulting solutions are guaranteed to satisfy the rank-one criteria. However, there is a major issue with regard to the convergence of this approach. To illustrate this, consider the optimization at the $0,1,2$-th iteration. Note that if the underlying optimization problem has a bounded value, which is the case in the problem considered in this paper, and the objective value achieved in successive iterations is nondecreasing, then the convergence is guaranteed. However, this monotonicity property cannot be said to be valid for this existing iterative approach, since the leftmost inequality is not always guaranteed in the following sequence of inequalities
\bea
f(\{\bold{u}_k^2\}) \geq f(\{\bold{u}_k^1\}) \geq f(\{\bold{u}_k^0\}),
\eea
where 
\bea
\!\!\!f(\{\bold{u}_k^0\})\!\! \triangleq \!\! |\alpha_0|^2\sum_{k=1}^K\!\! \left(\bold{u}_k^0\right)^{H}\!\! \bold{A}^H(\theta_0)\bold{G}^{-1}\!\!\left(\{\bold{u}_k^{0}\}\right)\!\!\bold{A}(\theta_0)\bold{u}_k^{0},
\eea
\bea
\!\!\!f(\{\bold{u}_k^1\})\!\! \triangleq \!\! |\alpha_0|^2\sum_{k=1}^K \!\! \left(\bold{u}_k^1\right)^{H}\!\! \bold{A}^H(\theta_0)\bold{G}^{-1}\!\!\left(\{\bold{u}_k^{0}\}\right)\!\! \bold{A}(\theta_0)\bold{u}_k^{1},
\eea
\bea
\!\!\!f(\{\bold{u}_k^2\})\!\! \triangleq \!\! |\alpha_0|^2\sum_{k=1}^K \!\! \left(\bold{u}_k^2\right)^{H}\!\! \bold{A}^H(\theta_0)\bold{G}^{-1}\!\!\left(\{\bold{u}_k^{1}\}\right)\!\!\bold{A}(\theta_0)\bold{u}_k^{2},
\eea
since the choice of $\{\bold{u}_k^2\}$ is based on maximizing
\bea
|\alpha_0|^2\sum_{k=1}^K \bold{u}_k^H \bold{A}^H(\theta_0)\bold{G}^{-1}\left(\{\bold{u}_k^{1}\}\right)\bold{A}(\theta_0)\bold{u}_k,
\eea
and not on maximizing 
\bea
|\alpha_0|^2\sum_{k=1}^K \bold{u}_k^H \bold{A}^H(\theta_0)\bold{G}^{-1}\left(\{\bold{u}_k^{0}\}\right)\bold{A}(\theta_0)\bold{u}_k.
\eea

Note that this violation of monotonocity occurs because one of the parameters of the optimization problem (specifically the matrix $\bold{G}^{-1}\left(\{\bold{u}_k^{i}\}\right)$) is different in each iteration. Since the optimization problem considered in this paper is a more generalized version of the problems considered in \cite{gc,lc,Ashraf1}, this monotonicity issue will also arise when the approaches proposed in those papers are applied to solve the SCNR maximization problem formulated in this paper.

In the following section, we present an efficient alternating optimization-based algorithm to solve \textbf{ P1}, which uses only a single receive beamformer for the detection process, and convergence is also guaranteed.

\section{Proposed Optimization Framework with Single Receive Beamformer}

In this section, we discuss the iterative optimization approach for solving the optimization problem \textbf{P1}. For ease of readability, we divide this section into two subsections. The first subsection discusses optimization with respect to $\bold{u}_k$'s for fixed $\bold{w}$; the second subsection discusses optimization with respect to $\bold{w}$ for fixed $\bold{u}_k$'s. 

As discussed in Section II, the coupling of the optimization variables makes it difficult to solve \textbf{P1}. To handle this difficulty, we use an alternating optimization approach, where within each iteration, the optimization is first performed on the information beamformers $\bold{u}_k$'s while fixing the receive combiner $\bold{w}$ and then the optimization is performed over receive beamformer $\bold{w}$ while fixing information beamformers $\bold{u}_k$'s. Before proceeding, we highlight an important property of the \textbf{P1}'s solution in the following lemma.
\begin{lemma}
    In the optimal solution of \textbf{P1} the power budget constraint is always met with equality. Mathematically, we have
    \bea
    \sum_{k=1}^K \|\bold{u}_k^*\|^2 = P_{max}.
    \eea
\end{lemma}
\begin{proof}
    Let us denote the minimum power required to satisfy the SINR constraints of CUs by $P_{min}$. Then, it can be shown that \textbf{P1} is feasible for any values of $P_{max} \geq P_{min}$. Next, using the contradiction, it can be easily shown that if $\sum_{k=1}^K\hat{\bold{u}}_k^H\hat{\bold{u}}_k = P_1 \geq P_{min}$ then the value of $\bar{\gamma}_r(\theta_i^T)$ $\forall~ i \in\{1,2, \cdots, I\}$ achieved by using $\hat{\bold{u}}_k$ is always lower than that achieved by choosing $\tilde{\bold{u}}_k$ such that $P_{max} \geq \sum_{k=1}^K\tilde{\bold{u}}_k^H\tilde{\bold{u}}_k = P_2 > P_1 \geq P_{min}$. Since the objective value does not decrease with each $\bar{\gamma}_r(\theta_i^T)$, we conclude that we must have $\sum_{k=1}^K \|\bold{u}_k^*\|^2 = P_{max}$ in the optimal solution. The proof is complete.  
\end{proof}

\subsection{Optimizing $\bold{u}_k$'s for Fixed $\bold{w}$}

For fixed value of $\bold{w}$ such that $\bold{w}^H\bold{w}=1$, \textbf{P1} can be written as follows:
\begin{maxi}|s|
{\bold{u}_k}{\min_{\theta_i^T \in \Theta^T} \frac{|\alpha_i^T|^2\sum_{k=1}^K \bold{u}_k^H\bold{\Phi}_i^T\bold{u}_k}{\sum_{j=1}^J|\alpha_j^C|^2\sum_{k=1}^K\bold{u}_k^H\bold{\Phi}_j^C\bold{u}_k+\sigma_r}}{}{}
\addConstraint{C1}
\addConstraint{\sum_{k=1}^K \|\bold{u}_k\|^2 = P_{max},} \label{16}
\end{maxi}
where 
\bea
\bold{\Phi}_i^T = \bold{A}^H(\theta_i^T)\bold{w}\bold{w}^H\bold{A}(\theta_i^T), \label{adef}
\eea
\bea
\bold{\Phi}_j^C = \bold{A}^H(\theta_j^C)\bold{w}\bold{w}^H\bold{A}(\theta_j^C).
\eea

Still, the problem (\ref{16}) is intractable due to minimization over the quadratic fractional terms in the objective function. To address this issue, we use an iterative approach. Toward this direction, we introduce the following notations
\begin{align}
\Tilde{\bold{\Phi}}_i^T = \bold{I}_{K}\otimes |\alpha_i^T|^2 \bold{\Phi}_i^T,
\end{align}
\vspace{-6mm}
\begin{align}
\bold{\Phi}^C = \sum_{j=1}^J |\alpha_j^C|^2 \bold{\Phi}_j^C,
\end{align}
\vspace{-6mm}
\begin{align}
\Tilde{\bold{\Phi}} = \bold{I}_{K} \otimes \bold{\Phi}^C + \frac{N_r}{P}\bold{I}_{KN_t}, \label{bdef}
\end{align}
\vspace{-6mm}
\begin{align}
\Tilde{\bold{H}}_k^s = \diag(\bold{e}_k) \otimes \bold{H}_k,
\end{align}
\vspace{-6mm}
\begin{align}
\Tilde{\bold{H}}_k^{int} = \left(\bold{I}_{K}-\diag({\bold{e}_k})\right) \otimes \bold{H}_k + \frac{N_0}{P} \bold{I}_{KN_t},
\end{align}
\vspace{-6mm}
\begin{align}
\bold{u} = \left[\vect({\bold{u}}_1),\cdots, \vect({\bold{u}}_K)\right].
\end{align}

Then, it can be shown that the problem (\ref{16}) can be equivalently written as:
\begin{maxi}|s|
{\bold{u}}{\min_{\theta_i^T \in \Theta^T} \frac{\bold{u}^H\Tilde{\bold{\Phi}}_i^T\bold{u}}{\bold{u}^H\Tilde{\bold{\Phi}}^C\bold{u}}}{}{}
\addConstraint{\bold{u}^H\bold{u}=P}
\addConstraint{\bold{u}^H\Tilde{\bold{H}}_k^s\bold{u}\geq \gamma_k \bold{u}^H\Tilde{\bold{H}}_k^{int}\bold{u}.}
\end{maxi}

By introducing a large penalty factor $\eta \to \infty$, the above problem can be replaced by the following alternative\footnote{For a very large value of $\eta$ the term $\sum_{i=1}^I\!\left(\|\Tilde{\bold{\Phi}}_i^{T\frac{1}{2}}\bold{u}\|_2\!-\!\sqrt{\lambda_i}\|\Tilde{\bold{\Phi}}^{C\frac{1}{2}}\bold{u}\|_2\!\!\right)^2$ becomes zero for the optimal solution, which is equivalent to $\lambda_i  = \frac{\bold{u}^H\Tilde{\bold{\Phi}}_i^T\bold{u}}{\bold{u}^H\Tilde{\Phi}^C\bold{u}}, \forall i \in \{1,2,\cdots,I\}$.} problem
\begin{maxi}|s|
{\bold{u}, \lambda_i}{\min_{i \in \{1,\cdots, I\}}\!\! \{\lambda_i\}\!-\!\eta \! \sum_{i=1}^I\!\left(\|\Tilde{\bold{\Phi}}_i^{T\frac{1}{2}}\bold{u}\|_2\!-\!\sqrt{\lambda_i}\|\Tilde{\bold{\Phi}}^{C\frac{1}{2}}\bold{u}\|_2\!\!\right)^2}{}{}
\addConstraint{\bold{u}^H\bold{u}=P}
\addConstraint{\bold{u}^H\Tilde{\bold{H}}_k^s\bold{u}\geq \gamma_k \bold{u}^H\Tilde{\bold{H}}_k^{int}\bold{u}}
\addConstraint{\lambda_i \geq 0.}\label{17}
\end{maxi}

It can be shown that (\ref{17}) is equivalent to the following problem
\begin{maxi}|s|
{\bold{u}, \lambda_i, \bold{Q}_i}{\min_{i \in \{1,\cdots, I\}}\!\! \{\lambda_i\}\!-\!\eta \! \sum_{i=1}^I\!\|\Tilde{\bold{\Phi}}_i^{T\frac{1}{2}}\bold{u}\!-\!\sqrt{\lambda_i}\bold{Q}_i\Tilde{\bold{\Phi}}^{C\frac{1}{2}}\bold{u}\|^2}{}{}
\addConstraint{\bold{u}^H\bold{u}=P}
\addConstraint{\bold{u}^H\Tilde{\bold{H}}_k^s\bold{u}\geq \gamma_k \bold{u}^H\Tilde{\bold{H}}_k^{int}\bold{u}}
\addConstraint{\lambda_i \geq 0}
\addConstraint{\bold{Q}_i^H\bold{Q}_i=\bold{I}_{KN_t}.}
\end{maxi}

The equivalence between (32) and (33) is guaranteed since for any values of $\bold{u},\lambda_i$, the optimal value of $\bold{Q}_i$ aligns the vector $\Tilde{\bold{\Phi}}^{C\frac{1}{2}}\bold{u}$ in the same direction as $\Tilde{\bold{\Phi}}_i^{T\frac{1}{2}}\bold{u}$ \cite{gnp}. Next, we discuss the alternating optimization approach for solving (33) with respect to $\bold{u}$, $\bold{Q}_i$ and $\lambda_i$.
\subsubsection{Optimization with respect to $\bold{u}$}
The optimization problem with respect to $\bold{u}$ for fixed values of $\lambda_i,\bold{Q}_i$ can be written as
\begin{maxi}|s|
{\bold{u}}{-\!\eta  \bold{u}^H\bold{R}\bold{u}}{}{}
\addConstraint{\bold{u}^H\bold{u}=P}
\addConstraint{\bold{u}^H\Tilde{\bold{H}}_k^s\bold{u}\geq \gamma_k \bold{u}^H\Tilde{\bold{H}}_k^{int}\bold{u},} \label{18}
\end{maxi}
where
\begin{align}
\bold{R}\! =\!\! \sum_{i=1}^I\!\left(\Tilde{\bold{\Phi}}_i^T\!\!+\!\lambda_i \Tilde{\bold{\Phi}}^C\!\!-\!\sqrt{\lambda_i}\left(\Tilde{\bold{\Phi}}_i^{T\frac{1}{2}}\bold{Q}_i\Tilde{\bold{\Phi}}^{C\frac{1}{2}}\!\!+\Tilde{\bold{\Phi}}^{C\frac{1}{2}}\bold{Q}_i^H\Tilde{\bold{\Phi}}_i^{T\frac{1}{2}}\right)\right).
\end{align}
Note that the matrix $\bold{R}$ is not necessarily a semidefinite matrix. To convert it to a semidefinite matrix, we define a matrix $\hat{\bold{R}}$ such that
\bea
\hat{\bold{R}} = \mu \bold{I}_{KN_t} - \bold{R},
\eea
where $\mu>0$ is chosen such that
\bea
\hat{\bold{R}} \succeq \bold{0}.
\eea

Next, we use the observation highlighted in Lemma 1 that $\sum_{k=1}^K\|\bold{u}_k^*\|^2=\|\bold{u}^*\|^2=P$. Hence, problem (\ref{18}) can be equivalently written as follows:
\begin{maxi}|s|
{\bold{u}}{\!  \bold{u}^H\hat{\bold{R}}\bold{u}}{}{}
\addConstraint{\bold{u}^H\bold{u}=P}
\addConstraint{\bold{u}^H\Tilde{\bold{H}}_i^s\bold{u}\geq \gamma_i \bold{u}^H\Tilde{\bold{H}}_i^{int}\bold{u}.} \label{19}
\end{maxi}

It is clear that the problem (\ref{19}) is a non-convex optimization problem, and therefore it is not possible to obtain a global optimal solution. To solve this problem, we use the majorization-minorization approach \cite{MM,MM1,MM3}. Specifically, we employ an iterative procedure to solve (\ref{19}). In this context, we note that for any convex quadratic function $f(\bold{x})= \bold{x}^H\bold{A}\bold{x}$ with $\bold{A}\succeq \bold{0}$, at any point $\bold{x}_0$ we have
\bea
f(\bold{x}) \geq 2 \Re(\bold{x}^H\bold{A}\bold{x}_0) - \bold{x}_0^H\bold{A}\bold{x}_0 \triangleq \hat{f}(\bold{x},\bold{x}_0)
\eea
 for any $\bold{x}_0$. Furthermore, at $\bold{x}=\bold{x}_0$ we have \bea
f(\bold{x}_0) = \hat{f}(\bold{x}_0,\bold{x}_0).
 \eea

Furthermore, it is easy to show that $\hat{f}(\bold{x},\bold{x}_0)$ satisfies the following gradient equality at $\bold{x}=\bold{x}_0$
 \bea
 \Delta f(\bold{x})|_{\bold{x}_0} = \Delta \hat{f}(\bold{x},\bold{x}_0)|_{\bold{x}=\bold{x}_0}
 \eea

Hence, $\hat{f}(\bold{x},\bold{x}_0)$ is a valid surrogate function for $f(\bold{x})$. With the help of the surrogate function, during the $s$-th inner iteration, we replace the objective function and the left-hand side of SINR constraints with their respective surrogate functions as follows:
\bea
\bold{u}^H\hat{\bold{R}}\bold{u} \geq  2\Re(\bold{u}^H\hat{\bold{R}}\bold{u}_{d,s-1})-\bold{u}_{d,s-1}^H\hat{\bold{R}}\bold{u}_{d,s-1}, \\
\bold{u}^H\Tilde{\bold{H}}_i^s\bold{u} \geq  2 \Re(\bold{u}^H\Tilde{\bold{H}}_i^s\bold{u}_{d,s-1}) - \bold{u}_{d,s-1}^H\Tilde{\bold{H}}_i^s\bold{u}_{d,s-1},
\eea
where $\bold{u}_{d,s-1}$ is the optimal value achieved at the end of $(s-1)$-th inner iteration within $d$-th outer iteration. Although by substituting the above surrogate functions, the objective function and the SINR constraints have become convex, the power budget constraint is still nonconvex. To deal with this constraint, we use the penalty-based method. In particular, we introduce a new non-negative optimization variable $b$ and replace the power budget constraint with the following constraints
\bea
\bold{u}^H\bold{u} \leq P + b,
\eea
\bea
\bold{u}^H\bold{u}\! \geq\! P\!-\!b \simeq 2 \Re(\bold{u}^H\bold{u}_{d,s-1})-\|\bold{u}_{d,s-1}\|_2^2 \geq P-b.
\eea

It can be seen that the constraints $\bold{u}^H\bold{u} \leq P + b$ and $\bold{u}^H\bold{u}\! \geq\! P\!-\!b$ together with $b\geq 0 $ imply $\bold{u}^H\bold{u} = P$. Note also that, when $b=0$ the constraints $\bold{u}^H\bold{u} \leq P + b$, and $2 \Re(\bold{u}^H\bold{u}_{d,s-1})-\|\bold{u}_{d,s-1}\|_2^2 \geq P-b$ imply $\bold{u}^H\bold{u}=P$. Therefore, in order to ensure that $b \to 0$ is the optimal solution, we introduce a penalty term, $\nu b$, in the objective function. With these modifications, problem (\ref{19}) can be relaxed to the following optimization problem.
\begin{maxi}|s|
{\bold{u},b}{\!  2\Re(\bold{u}^H\hat{\bold{R}}\bold{u}_0)-\bold{u}_0^H\hat{\bold{R}}\bold{u}_0 -\nu b }{}{}
\addConstraint{\bold{u}^H\bold{u}\leq P + b}
\addConstraint{2 \Re(\bold{u}^H\bold{u}_0)-\|\bold{u}_0\|_2^2\geq P-b}
\addConstraint{b\geq 0}
\addConstraint{2 \Re(\bold{u}^H\Tilde{\bold{H}}_i^s\bold{u}_0) - \bold{u}_0^H\Tilde{\bold{H}}_i^s\bold{u}_0 \geq \gamma_i \bold{u}^H\Tilde{\bold{H}}_i^{int}\bold{u},} \label{upro}
\end{maxi}
where $\nu$ is a very large positive number. Note that problem (\ref{upro}) is a convex optimization problem and can be easily solved with the help of off-the-shelf optimization tools such as CVX.

\subsubsection{Optimization with respect to $\bold{Q}_i$'s}

For fixed values of $\bold{u},\lambda_i$, problem (30) can be converted into subproblems $I$, where the subproblem $i$-th is given as
\begin{mini}|s|
    {\bold{Q}_i}{\|\Tilde{\bold{\Phi}}_i^{T\frac{1}{2}}\bold{u}\!-\!\sqrt{\lambda_i}\bold{Q}_i\Tilde{\bold{\Phi}}^{C\frac{1}{2}}\bold{u}\|^2}{}{}
    \addConstraint{\bold{Q}_i^H\bold{Q}_i=\bold{I}.}
\end{mini}
The optimal solution for (47) is given as \cite{gnp}
\bea
\bold{Q}_i = \bold{Q}_{\bold{a}_i}\bold{Q}_{\bold{b}_i}^H, \label{qs}
\eea
where 
\bea
\bold{Q}_{\bold{a}_i} = [\bold{a}_{i1},\bold{a}_{i2},\cdots,\bold{a}_{iKN_t}],
\eea
\bea
\bold{Q}_{\bold{b}_i} = [\bold{b}_{i1},\bold{b}_{i2},\cdots,\bold{b}_{iKN_t}],
\eea
\bea
\bold{a}_{i1} = \frac{\Tilde{\bold{\Phi}}_i^{T\frac{1}{2}}\bold{u}}{\|\Tilde{\bold{\Phi}}_i^{T\frac{1}{2}}\bold{u}\|_2}, \bold{b}_{i1} = \frac{\Tilde{\bold{\Phi}}^{C\frac{1}{2}}\bold{u}}{\|\Tilde{\bold{\Phi}}^{C\frac{1}{2}}\bold{u}\|_2}, 
\eea
and
\bea
\bold{a}_{il}^H\bold{a}_{im} = \begin{cases}
  1  & \text{ if } l= m, \\
  0 & \text{ if } l \neq m,
\end{cases}
\eea
\bea
\bold{b}_{il}^H\bold{b}_{im} = \begin{cases}
  1  & \text{ if } l= m, \\
  0 & \text{ if } l \neq m.
\end{cases}
\eea

Note that after obtaining the values of $\bold{a}_{i1},\bold{b}_{i1}$ from the closed-form expressions in (51), the remaining orthonormal vectors, $\bold{a}_{il},\bold{b}_{il}~ \forall ~l\in\{2,\cdots,KN_t\}$, can be obtained by using the well-known Gram-Schmidt orthogonalization procedure. 

\subsubsection{Optimization with respect to $\lambda_i$'s}

For a fixed value of $\bold{u}$, the optimization problem with respect to $\lambda_i$'s can be written as
\begin{maxi}|s|
{\lambda_i}{\min_{i \in \{1,\cdots, I\}}\!\! \{\lambda_i\}\!-\!\eta \! \sum_{i=1}^I\!\left(x_i\!-\!\sqrt{\lambda_i}y\right)^2}{}{}
\addConstraint{\lambda_i\geq 0,} \label{lampro}
\end{maxi}
where 
\bea
x_i = \|\Tilde{\bold{\Phi}}_i^{T\frac{1}{2}}\bold{u}\|_2, y = \|\Tilde{\bold{\Phi}}^{C\frac{1}{2}}\bold{u}\|_2.
\eea

Since the point-wise minimum of concave functions is also a concave function, we conclude that (54) is a convex optimization problem. Therefore, a global optimal solution for (54) can be obtained using CVX. 

\subsection{Optimizing $\bold{w}$ for Fixed $\bold{u}_k$'s}

For fixed values of $\bold{u}_k$'s, the SCNR maximization problem can be written as 
\begin{maxi}|s|
    {\bold{w}}{\min_{\theta_i^T \in \Theta^T} \frac{|\alpha_i^T|^2\bold{w}^H\sum_{k=1}^K\bold{B}^T_{i,k}\bold{w}}{\bold{w}^H\left[\sum_{j=1}^J|\alpha_j^C|^2\sum_{k=1}^K\bold{B}^C_{j,k}+\bold{I}\right]\bold{w}}}{}{}
    \addConstraint{\|\bold{w}\|^2=1,}
\end{maxi}
where 
\bea
\bold{B}^l_{m,k}=\bold{A}(\theta_m^l) \bold{u}_k\bold{u}_k^H\bold{A}^H(\theta_m^l), \label{bs}
\eea
$~\forall~ l \in \{C,T\},~ \forall ~ m \in \{i,j\}, ~\forall k \in \{1,\cdots,K\}.$ By introducing a new matrix variable $\bold{W}=\bold{ww}^H$, (56) can be equivalently written as
\begin{maxi}|s|
    {\bold{W}}{\min_{\theta_i^T \in \Theta^T} \frac{|\alpha_i^T|^2\Tr\left(\bold{W}^H\sum_{k=1}^K\bold{B}^T_{i,k}\right)}{\Tr\left(\bold{W}^H\left[\sum_{j=1}^J|\alpha_j^C|^2\sum_{k=1}^K\bold{B}^C_{j,k}+\bold{I}\right]\right)}}{}{}
    \addConstraint{\Tr\left(\bold{W}\right)=1}
    \addConstraint{\rank\left(\bold{W}\right)=1.}
\end{maxi}

It is clear that (58) is a nonconvex optimization problem due to the non-convex rank constraint and the fractional form of the objective function. However, it can be shown that each of the matrices $\bold{B}_{m,k}^l$ is a Toeplitz matrix \cite{Ashraf1}. This means that (58) is a special case of the generalized fractional program with Toeplitz quadratics. For this class of problems, it was shown in \cite{dem} that the global optimal solution can be obtained with Dinkleback's algorithm in polynomial time. 

\section{Proposed Iterative Algorithm}
Based on the above discussion, we propose an iterative algorithm to solve \textbf{P1}. The proposed iterative algorithm is shown as \textbf{Algorithm 1} at the top of next page. First, the initialization is done in lines 1-2. Then, the outer loop starts from line 3. The alternating optimization with respect to $\bold{u},\bold{Q}_i,\lambda_i$, and $\bold{w}$ is done within the outer loop (lines 3-24). Then, an inner loop (lines 4-13) is used to perform optimization with respect to $\bold{u}$. Optimization with respect to $\bold{Q}_i,\lambda_i$ and $\bold{w}$ is carried out in lines 15, 16, and 18, respectively. Next, a check for the convergence is performed on line 20. Finally, the acquired optimal solutions are output in line 25. In the following, we discuss the convergence of \textbf{Algorithm 1}.
\begin{algorithm}[t]
\renewcommand{\thealgorithm}{1:}
\caption{Proposed iterative alternating optimization algorithm for solving \textbf{P1}.}\label{euclid_1}
\begin{algorithmic}[1]
\State Initialize $s=0,d=1, \epsilon, \lambda_i, \bold{Q}_i, \eta, \mu, \nu$, $d_{max}$, $S$ $\bold{w}$,$ \bold{u}_k, \forall~ k \in \{1, \cdots, K\}$, $\bold{A}(\theta_i^T), \bold{A}(\theta_j^C), \forall~ i \in \{1,\cdots, M\}, \forall~ j \in \{1,\cdots,J\}$.
\State Initialize $\bold{u}_{0,S}$ through some communication-based resource allocation scheme.
\While{$d \leq d_{max}$}
\While{$s < S+1$}
\If{$s=0$}
\State $\bold{u}_0=\bold{u}_{d,0}=\bold{u}_{d-1,S}$.
\Else 
\State $\bold{u}_0 = \bold{u}_{d,s-1}$.
\EndIf
\State Solve problem (\ref{upro}) with respect to $\bold{u}$ for fixed values of $\bold{w}, \lambda_i, \bold{Q}_i$ and obtain the values of $\bold{u}_k$'s from $\bold{u}$.
\State Store the optimal solution as $\bold{u}_{d,s}$.
\State $s = s+1$.
\EndWhile
\State $s=0$.
\State Obtain the values of $\bold{Q}_i$'s using (\ref{qs}).
\State Solve problem (\ref{lampro}) and obtain the values of $\lambda_i$.
\State Obtain the values of $\bold{B}_{m,k}^l$'s through (\ref{bs}).
\State Solve problem (58) and denote the solution obtained with $\bold{w}^d$.
\State Store the objective value obtained as $\chi\left(\bold{w}^d,{\{\bold{Q}_c^d\},\!\bold{u}_d,\!\{\lambda_i^d\}}\right)$.
\If{$\bigg|\chi\left({\bold{w}^{d},\{\bold{Q}_i^{d}\},\!\bold{u}_{d},\!\{\lambda_i^{d}\}}\right)-\chi\left({\bold{w}^{d-1},\{\bold{Q}_i^{d-1}\},\!\bold{u}_{d-1},\!\{\lambda_i^{d-1}\}}\right)\bigg|\leq \epsilon$}
\State Break
\EndIf
\State Set $d = d + 1$.
\EndWhile
\State Output $\bold{u}_k^*$, $\bold{w}^*$, where $\bold{u}_k^*= \bold{u}_k^{d-1}$, $\bold{w}^*=\bold{w}^{d-1}$.
\end{algorithmic}
\end{algorithm}

Next, we discuss the convergence of \textbf{Algorithm 1}. The proposed algorithm has an outer loop, and it is clear that the algorithm optimizes the optimization variables in an alternating manner within each outer iteration. Let us denote by $\chi(\bold{w}^d,\{\bold{Q}_i^d\},\bold{u}_d,\{\lambda_i^d\})$ the objective value achieved at the end of the $d$-th outer iteration. Mathematically,
\bea
&\chi\bigg(\bold{w}^{d},\{\bold{Q}_i^d\},\!\bold{u}_d,\!\{\lambda_i^d\}\bigg)\! \nonumber \\& \triangleq \!\! \min_{i \in \{1,\cdots, I\}}\!\! \{\lambda_i^d\}\!-\!\eta \! \sum_{i=1}^I\!\|\Tilde{\bold{\Phi}}_i^{T\frac{1}{2}}\bold{u}_d\!-\!\sqrt{\lambda_i^d}\bold{Q}_i^d\Tilde{\bold{\Phi}}^{C\frac{1}{2}}\bold{u}_d\|^2.
\eea

Note that the dependence of $\chi\left(\{\bold{Q}_i^d\},\!\bold{u}_d,\!\{\lambda_i^d\}\right)$ on $\bold{w}$ is hidden within $\Tilde{\bold{\Phi}}_i^{T\frac{1}{2}}$ and $\Tilde{\bold{\Phi}}^{C\frac{1}{2}}$ as can be checked through (\ref{adef})-(\ref{bdef}). Then, to prove the convergence of the proposed algorithm, we must show that the following condition is always satisfied
\bea
\chi\left(\{\bold{Q}_i^{d}\},\bold{u}_{d},\{\lambda_i^{d}\}\right) \geq \chi\left(\{\bold{Q}_i^{d-1}\},\bold{u}_{d-1},\{\lambda_i^{d-1}\}\right).
\eea

As detailed above, the optimization problems with respect to $\{\bold{Q}_i\}$, $\{\lambda_i\}$ and $\bold{w}$ can be solved globally optimally. Therefore, the objective value achieved as a result of optimization over $\{\bold{Q}_i\}$, $\{\lambda_i\}$, and $\bold{w}$ is always nondecreasing. Mathematically, this statement can be written as
\bea
\chi\!\left(\bold{w}^{d}\!,\!\{\bold{Q}_i^{d}\},\!\bold{u}_{d\!-\!1},\!\{\lambda_i^{d}\}\right)\!\! \geq \!\!\chi\!\left(\bold{w}^{d\!-\!1}\!,\!\{\bold{Q}_i^{d\!-\!1}\},\!\bold{u}_{d\!-\!1},\!\{\lambda_i^{d\!-\!1}\}\!\right).
\eea

Next, with respect to $\bold{u}$, we note that for a very large value of the penalty factor, i.e. $\nu \to \infty$, the optimal value of $b$ in (\ref{upro}) will approach zero, that is, $b^* \to 0$. Note also that (\ref{lampro}) is a convex optimization problem for which a global optimal solution can be achieved. This means that the objective value achieved at the end of each inner iteration is always nondecreasing. Now combining this with the following facts: (i) the objective value of (\ref{upro}) is bounded from above, (ii) $b^* \to 0$, and (iii) the equality of the surrogate function, we conclude that $\bold{u}_{d,s}^H\hat{\bold{R}}\bold{u}_{d,s}\geq \bold{u}_{d,s-1}^H\hat{\bold{R}}\bold{u}_{d,s-1}$. Since we use $\bold{u}_{d-1,S}=\bold{u}_{d}=\bold{u}_{d,0}$, it can be established that 
\bea
\chi\!\left(\bold{w}^d\!,\!\{\bold{Q}_i^{d}\},\bold{u}_{d},\!\{\!\lambda_i^{d}\}\!\right)\! \geq \! \chi\!\left(\bold{w}^{d\!-\!1}\!,\!\{\bold{Q}_i^{d\!-\!1}\},\!\bold{u}_{d\!-\!1},\!\{\!\lambda_i^{d-1}\}\!\right).
\eea

Furthermore, since the problem has a bounded objective value, the monotonicity of the achieved objective value in each iteration guarantees the convergence of \textbf{Algorithm 1}.

\section{Numerical Results}

In this section, we present the numerical results. Unless otherwise specified, the important parameters of the system are provided in Table 1. In the following, we illustrate four types of numerical results. First, we show the SCNR performance improvement achieved by the proposed algorithm. Second, we show the convergence result of the proposed iterative algorithm. Third, we show the effect of the number of possible directions of the target and its angular spread on the achieved SCNR. Finally, a comparison between the receiver's beampattern obtained from the proposed algorithm and those obtained by using dedicated receive beamformers for each possible direction is shown.

\begin{table}
\caption{System parameters.}
\begin{center}
\begin{tabular}{ |c|c|c|c|} 
 \hline
 Parameter & value & Parameter & value\\ 
 \hline
 \hline
 $f_c$ & $30$ GHz & Bandwidth & $100$ MHz \\
 \hline
 $N_0$ & $-94$ & Propagation & UMi\\ 
 \hline
 $N_t$ & $\{8,12,16\}$ & $N_R$ & $\{8,12,16\}$ \\
 \hline
 $K$ & $4$ & $\Gamma_k$ & $[0,20]$ dB \\ 
 \hline
 $P_{max}$ & $30$ dBm & $d_1^{CU}$ & $10$ m \\
 \hline
 $d_2^{CU}$ & $15$ m & $d_3^{CU}$ & $20$ m \\
 \hline
 $d_4^{CU}$ & $25$ m & $J$ & $2$ \\
 \hline
 $|\alpha_1^C|^2$ & $.001$ & $|\alpha_2^C|^2$ & $.00001$ \\
 \hline
 $\theta_T$ & $\{\pi/4,\pi/6\}$ & $\theta_1^C$ & $0$ \\
 \hline
 $\theta_2^C$ & $\frac{\pi}{2}$ & $\alpha_i^T$ & $10^{-1.5}$ \\
 \hline
\end{tabular}
\end{center}
\end{table}

\subsection{SCNR Results}

Fig. 2 shows the SCNR results obtained by the proposed algorithm for different numbers of transmit/receive antennas with respect to CU SINR thresholds. It can be seen that as the SCNR decreases, the CU SINR thresholds increase. This behavior can be attributed to the fact that a higher CU SINR threshold will require a higher transmit beamforming gain toward the CUs, thus resulting in lower levels of transmit powers being directed toward possible target directions. This results in lower reflected power from the target directions, which ultimately causes a reduction in the received SCNR. However, it can be noted that SCNR degradation is less severe for a larger number of antennas, especially for higher SINR thresholds. This is due to the fact that a larger number of antennas can provide more degrees of freedom in designing the beampatterns, which can support higher directivity gains in multiple directions. It is worth mentioning here that while we have shown results in terms of SINR for communication system, the corresponding data rates can be easily obtained by using the well-known Shannon equation for capacity.

\begin{figure}
  \centering
  \includegraphics[width=\columnwidth]{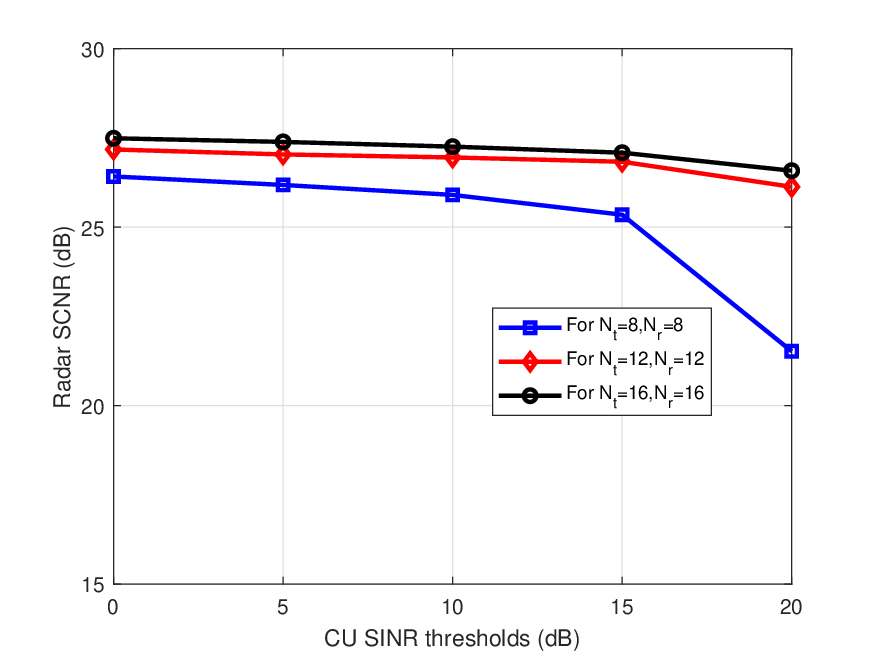}
  \caption{Achieved SCNR as a function of CU SINR thresholds for different numbers of transmit/receive antennas.}
\end{figure}

\subsection{Convergence Results}

Fig. 3 shows the convergence results of the proposed algorithm for different values of the CU SINR thresholds. It can be seen that the performance almost converges only after $3$ iterations for each realization of the CU channels. Therefore, the proposed algorithm does not involve high computational complexity due to the need for a greater number of iterations to achieve convergence. This result validates the discussion carried out in Section IV, where the theoretical guarantee for the convergence is proved by using the argument that the nondecreasing value of objective function is achieved at the end of each iteration.

\begin{figure}
  \centering
  \includegraphics[width=\columnwidth]{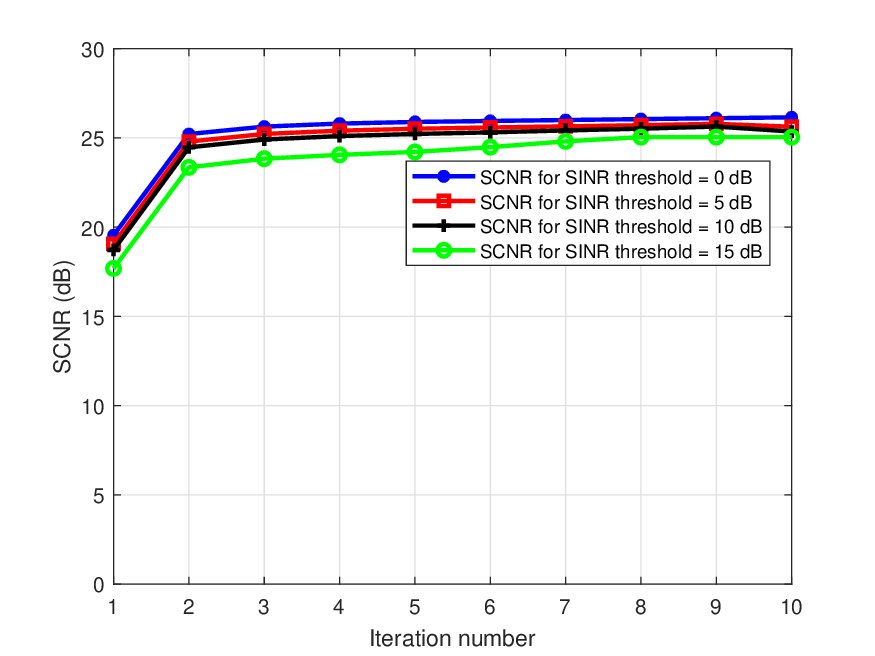}
  \caption{Convergence of the proposed algorithm for different values of CU SINR thresholds.}
\end{figure}

\subsection{SCNR for Different $I$ and the Angular Spread of Target Direction}

An important parameter in the considered optimization problem is the number of possible directions for the target and the angular spread of the arrival direction. In Fig. 4, we show the SCNR achieved for different values of $I$. Note that the achieved SCNR values are plotted on a linear scale. Therefore, we observe that a change in the value of $I$ has an almost negligible effect on the SCNR achieved. Another important aspect is the angular spread/gap over which the target can arrive. To analyze its effect, we consider two possible angular gaps denoted by $G1$ and $G2$ where $G1$ is $15$-$50$ degrees angular gap while $G2$ is $15$-$40$ degrees angular gap. We can see from Fig. 4 that for $G1$, there is only a slight decrease in the SCNR values. This almost identical SCNR performance can be attributed to the almost identical receive beamformer designs, shown in Fig. 5, for both angular spreads. From Fig. 5, it can be observed that while there are some dissimilarities in the beampattern over undesired angles, the beampatterns are almost identical over the angular spread of target directions and the clutter angles. The results presented in Fig. 4 and Fig. 5 suggest that it is sufficient to use only the angular spread information for designing the receive beamformer by using $I=2$ and choosing $\theta_1^T,\theta_2^T$ to be the two extreme angles within the angular spread.   

\begin{figure}
  \centering
  \includegraphics[width=\columnwidth]{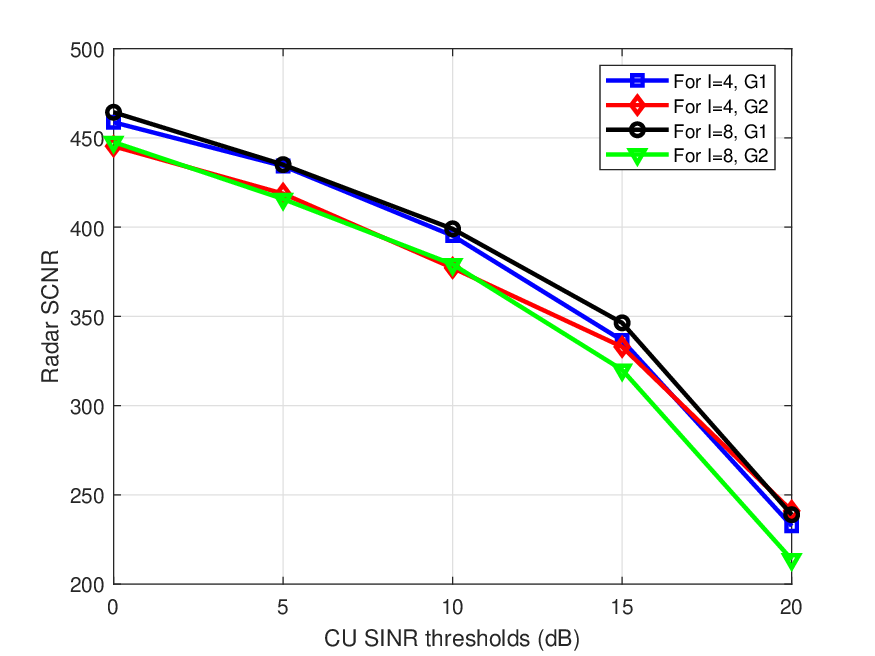}
  \caption{Achieved SCNR for different $I$ and the angular spread of target directions.}
\end{figure}

\begin{figure}
  \centering
  \includegraphics[width=\columnwidth]{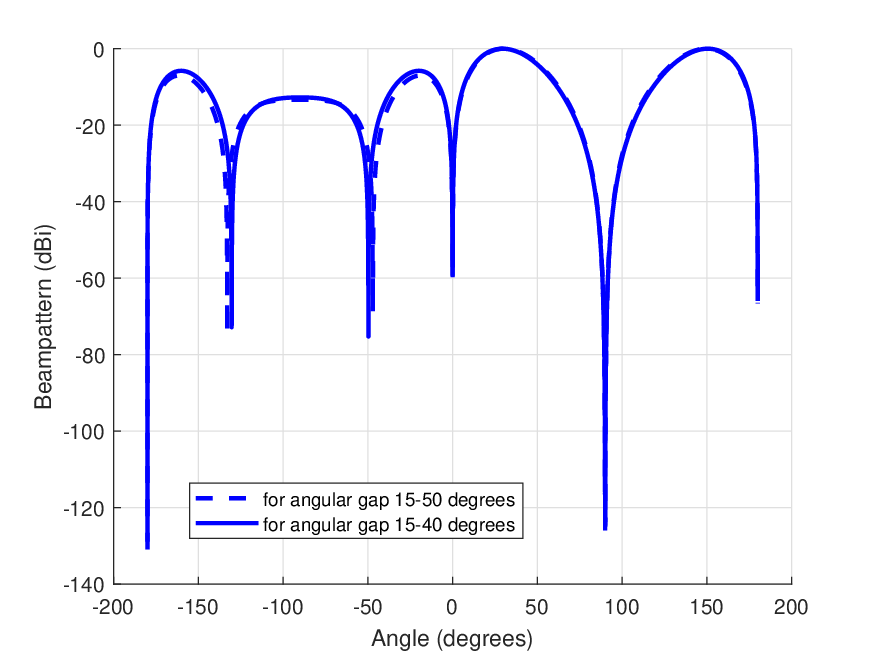}
  \caption{Receiver's beampattern for the different angular spread of target direction.}
\end{figure}
\subsection{Comparison Between Beampattern for 
Proposed Scheme and the Dedicated Beampatterns}

\begin{figure}
  \centering
  \includegraphics[width=\columnwidth]{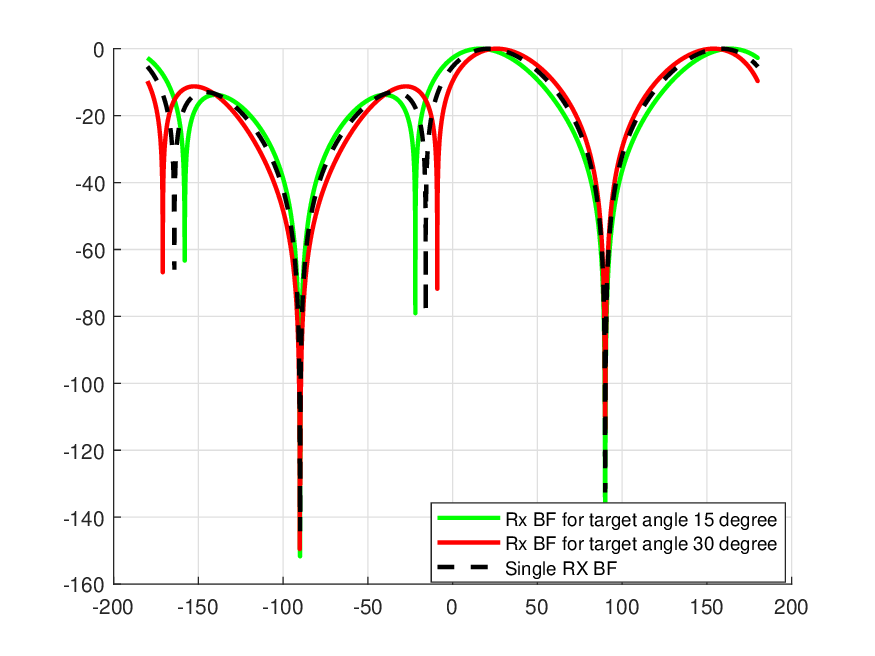}
  \caption{Receiver's beampattern for the proposed scheme and the dedicated beampatterns for different target directions.}
\end{figure}

Next, we compare the beam pattern of the proposed scheme with the dedicated beamformer scheme, which uses $I$ receive beamformers to combine at the receiver. For this comparison, we assume $\theta_i^T=15$ degree and $\theta_2^T=30$ degree. The clutter angles are assumed to be $\theta_1^C=90$ degree and $\theta_2^C=-90$ degree. The received beamformers obtained for both schemes are shown in Fig. 6. It can be seen that the dedicated receive beamformers provide the highest gains at the respective desired angles. On the other hand, the beampattern for the proposed scheme does not have the highest gain at either of the possible target directions, but it provides a higher gain over a range of possible target angles, which accounts for the "maximizing-the-minimum" aspect within the optimization problem. Although the beamforming gain achieved for the proposed scheme is inferior to the dedicated beamforming scheme at the exact possible target angles, the difference is almost negligible. Therefore, it is sufficient to use a single receive beamformer for combining, even when there is ambiguity in the arrival direction of the target.

\section{Conclusions}

An iterative optimization algorithm is proposed to solve the SCNR maximization problem with uncertainty in the target arrival direction. The optimization problem formulated with respect to transmitting beamformers is converted into a convex problem by first applying the penalty-based approach and then by using the MM approach. After fixing the transmit beamformers, it is observed that the resulting problem with respect to the receive beamformer is a special case of the generalized fractional program with Toeplitz quadratics, for which a global optimal solution can be obtained through the Dinkelback algorithm. It is also shown that the resulting objective value achieved after each iteration is nondecreasing, and hence the convergence of the proposed algorithm is also guaranteed. 

The different aspects of the performance of the proposed algorithm are then analyzed using numerical results. Regarding the proposed algorithm, two important conclusions can be drawn: (i) It can be concluded that it is not necessary to update the receive beamformer's weights if the number of possible target directions has increased with a fixed angular spread in the target arrival direction, and (ii) the receive beampattern remains almost unchanged with a slight change in the angular spread of the target direction.

%\section*{Acknowledgment}
%This work was supported by the Academy of Finland ACCESS project: Autonomous Communication Converged with Efficient Sensing for UAV Swarms, Project No. 339519, Business Finland Project: 5G-PSS, and H2020 MSCA RISE project: DIOR (Peoject No: 101008280, DIOR: Deep Intelligent Optical and Radio Communication Networks.

\bibliographystyle{IEEEtran}
\bibliography{refs}

% Generated by IEEEtran.bst, version: 1.14 (2015/08/26)
\begin{thebibliography}{10}
\providecommand{\url}[1]{#1}
\csname url@samestyle\endcsname
\providecommand{\newblock}{\relax}
\providecommand{\bibinfo}[2]{#2}
\providecommand{\BIBentrySTDinterwordspacing}{\spaceskip=0pt\relax}
\providecommand{\BIBentryALTinterwordstretchfactor}{4}
\providecommand{\BIBentryALTinterwordspacing}{\spaceskip=\fontdimen2\font plus
\BIBentryALTinterwordstretchfactor\fontdimen3\font minus \fontdimen4\font\relax}
\providecommand{\BIBforeignlanguage}[2]{{%
\expandafter\ifx\csname l@#1\endcsname\relax
\typeout{** WARNING: IEEEtran.bst: No hyphenation pattern has been}%
\typeout{** loaded for the language `#1'. Using the pattern for}%
\typeout{** the default language instead.}%
\else
\language=\csname l@#1\endcsname
\fi
#2}}
\providecommand{\BIBdecl}{\relax}
\BIBdecl

\bibitem{IMT}
F.~Dong, F.~Liu, Y.~Cui, S.~Lu, and Y.~Li, ``{Sensing as a Service in 6G Perceptive Mobile Networks: Architecture, Advances, and the Road Ahead},'' \emph{arXiv:2308.08185}, 2023.

\bibitem{sxu}
S.~Xu, L.~Wu, K.~Doğançay, and M.~Alaee-Kerahroodi, ``{A hybrid approach to optimal TOA-sensor placement with fixed shared sensors for simultaneous multi-target localization},'' \emph{IEEE Trans. Sig. Proc.}, vol.~70, pp. 1197--1212, 2022.

\bibitem{yrong}
Y.~Rong, A.~Aubry, A.~De~Maio, and M.~Tang, ``{Adaptive radar detection in Gaussian interference using clutter-free training data},'' \emph{IEEE Trans. Sig. Proc.}, vol.~70, pp. 978--993, 2022.

\bibitem{fliu}
F.~Liu, Y.-F. Liu, A.~Li, C.~Masouros, and Y.~C. Eldar, ``{Cramér-Rao bound optimization for joint radar-communication beamforming},'' \emph{IEEE Trans. Sig. Proc.}, vol.~70, pp. 240--253, 2022.

\bibitem{ivali}
I.~Valiulahi, C.~Masouros, A.~Salem, and F.~Liu, ``{Antenna selection for energy-efficient dual-functional radar-communication systems},'' \emph{IEEE Wirel. Commun. Lett.}, pp. 1--1, 2022.

\bibitem{lx}
X.~Li, F.~Liu, Z.~Zhou, G.~Zhu, S.~Wang, K.~Huang, and Y.~Gong, ``{Integrated sensing, communication, and computation over-the-air: MIMO beamforming design},'' \emph{arXiv:2201.12581}, 2022.

\bibitem{dm}
A.~De~Maio, S.~De~Nicola, Y.~Huang, Z.-Q. Luo, and S.~Zhang, ``{Design of phase codes for radar performance optimization with a similarity constraint},'' \emph{IEEE Trans. Sig. Proc.}, vol.~57, no.~2, pp. 610--621, 2009.

\bibitem{co}
C.~D. Ozkaptan, E.~Ekici, and O.~Altintas, ``{Adaptive waveform design for communication-enabled automotive radars},'' \emph{IEEE Trans. Wirel. Comm.}, pp. 1--1, 2021.

\bibitem{lc}
L.~Chen, F.~Liu, J.~Liu, and C.~Masouros, ``{Composite signalling for DFRC: Dedicated probing signal or not?}'' \emph{arXiv:2009.03528}, 2020.

\bibitem{Ashraf1}
M.~Ashraf, B.~Tan, D.~Moltchanov, J.~S. Thompson, and M.~Valkama, ``{Joint optimization of radar and communications performance in 6G cellular systems},'' \emph{IEEE Trans. Green Commun. Network.}, vol.~7, no.~1, pp. 522--536, 2023.

\bibitem{Lchen1}
L.~Chen, F.~Liu, W.~Wang, and C.~Masouros, ``{Joint radar-communication transmission: A generalized Pareto optimization framework},'' \emph{IEEE Trans. Sig. Proc.}, vol.~69, pp. 2752--2765, 2021.

\bibitem{Fan1}
F.~Liu, C.~Masouros, A.~Li, H.~Sun, and L.~Hanzo, ``{MU-MIMO communications with MIMO radar: From co-existence to joint transmission},'' \emph{IEEE Trans. Wirel. Commun.}, vol.~17, no.~4, pp. 2755--2770, 2018.

\bibitem{elder}
X.~Liu, T.~Huang, N.~Shlezinger, Y.~Liu, J.~Zhou, and Y.~C. Eldar, ``{Joint transmit beamforming for multiuser MIMO communications and MIMO radar},'' \emph{IEEE Trans. Sig. Proc.}, vol.~68, pp. 3929--3944, 2020.

\bibitem{Fan2}
F.~Liu, L.~Zhou, C.~Masouros, A.~Li, W.~Luo, and A.~Petropulu, ``{Toward dual-dunctional radar-communication systems: Optimal waveform design},'' \emph{IEEE Trans. Sig. Proc.}, vol.~66, no.~16, pp. 4264--4279, 2018.

\bibitem{Fan3}
X.~Hu, C.~Masouros, F.~Liu, and R.~Nissel, ``{MIMO-OFDM dual-functional radar-communication systems: Low-PAPR waveform design},'' 2021.

\bibitem{Palomar}
L.~Zhao and D.~P. Palomar, ``{Maximin joint optimization of transmitting code and receiving filter in radar and communications},'' \emph{IEEE Trans. Sig. Proc.}, vol.~65, no.~4, pp. 850--863, 2017.

\bibitem{Ashraf2}
\BIBentryALTinterwordspacing
M.~Ashraf, A.~Gaydamaka, D.~Moltchanov, A.~Mohammad, and B.~Tan, ``{Maximizing detection of target with multiple direction possibilities to support immersive communications in Metaverse},'' in \emph{Proc. of the 2nd Workshop on Integrated Sensing and Communications for Metaverse}, ser. ISACom '23.\hskip 1em plus 0.5em minus 0.4em\relax New York, NY, USA: ACM, 2023, p. 30–35. [Online]. Available: \url{https://doi.org/10.1145/3597065.3597452.}
\BIBentrySTDinterwordspacing

\bibitem{Cwen}
C.~Wen, Y.~Huang, and T.~N. Davidson, ``{Efficient transceiver design for MIMO dual-function radar-communication systems},'' \emph{IEEE Trans. Sig. Proc.}, vol.~71, pp. 1786--1801, 2023.

\bibitem{gc}
G.~Cui, H.~Li, and M.~Rangaswamy, ``{MIMO radar waveform design with constant modulus and similarity constraints},'' \emph{IEEE Trans. Sig. Proc.}, vol.~62, no.~2, pp. 343--353, 2014.

\bibitem{gnp}
\BIBentryALTinterwordspacing
A.~Gharanjik, M.~Soltanalian, M.~R.~B. Shankar, and B.~Ottersten, ``{Grab-n-Pull: A max-min fractional quadratic programming framework with applications in signal and information processing},'' \emph{Sig. Proc.}, vol. 160, pp. 1--12, 2019. [Online]. Available: \url{https://www.sciencedirect.com/science/article/pii/S0165168419300556}
\BIBentrySTDinterwordspacing

\bibitem{MM}
\BIBentryALTinterwordspacing
D.~R. Hunter and K.~Lange, ``{A Tutorial on MM Algorithms},'' \emph{The American Statistician}, vol.~58, no.~1, pp. 30--37, 2004. [Online]. Available: \url{https://doi.org/10.1198/0003130042836}
\BIBentrySTDinterwordspacing

\bibitem{MM1}
\BIBentryALTinterwordspacing
K.~Lange, D.~R. Hunter, and I.~Yang, ``Optimization transfer using surrogate objective functions,'' \emph{J. Comp. Graph. Stat.}, vol.~9, no.~1, pp. 1--20, 2000. [Online]. Available: \url{http://www.jstor.org/stable/1390605}
\BIBentrySTDinterwordspacing

\bibitem{MM3}
P.~Stoica and Y.~Selen, ``Cyclic minimizers, majorization techniques, and the expectation-maximization algorithm: a refresher,'' \emph{IEEE Sig. Proc. Mag.}, vol.~21, no.~1, pp. 112--114, 2004.

\bibitem{dem}
A.~Aubry, V.~Carotenuto, and A.~D. Maio, ``{New results on generalized fractional programming problems with Toeplitz quadratics},'' \emph{IEEE Sig. Proc. Lett.}, vol.~23, no.~6, pp. 848--852, 2016.

\end{thebibliography}
\end{document}